\documentclass{llncs}
\usepackage{amssymb}
\usepackage{amsmath,fancyhdr}
\usepackage{amsfonts}
\usepackage{wrapfig}
\usepackage{epsfig}
\usepackage{enumerate}
\usepackage{verbatim}
\usepackage{graphics}
\usepackage{hyperref}
\usepackage{makeidx}
\usepackage{microtype}

\usepackage{lineno}

\def\cro{\mathop{\rm cr}\nolimits}

\def\qed{ \ \vrule width.2cm height.2cm depth0cm\smallskip}

\spnewtheorem{observation}{Observation}{\bfseries}{\rmfamily}

\begin{document}
\title{Hanani-Tutte for Radial Planarity II~\thanks{Missing proofs can be found in the appendices.}}

\author{Radoslav Fulek\inst{1}\thanks{The research leading to these results has received funding from the People Programme (Marie Curie Actions) of the European Union's Seventh Framework Programme (FP7/2007-2013) under REA grant agreement no [291734].} and Michael Pelsmajer\inst{2} and Marcus Schaefer\inst{3}}

\institute{
IST Austria, Am Campus 1, Klosterneuburg 3400, Austria\\
\email{radoslav.fulek@gmail.com}
\and
Illinois Institute of Technology, Chicago, Illinois 60616, USA\\
\email{pelsmajer@iit.edu}
\and
DePaul University, Chicago, IL 60604, USA\\
\email{mschaefer@cs.depaul.edu}
}

\maketitle

\begin{abstract}
A drawing of a graph $G$ is {\em radial} if the vertices of $G$ are placed on concentric circles $C_1, \ldots, C_k$ with common center $c$, and edges are drawn {\em radially}: every edge intersects every circle centered at $c$ at most once. $G$ is {\em radial planar} if it has a radial embedding, that is, a crossing-free radial drawing. If the vertices of $G$ are ordered or partitioned into ordered levels (as they are for leveled graphs), we require that the assignment of vertices to circles corresponds to the given ordering or leveling. A pair of edges $e$ and $f$ in a graph is {\em independent} if
$e$ and $f$ do not share a vertex.

We show that a graph $G$ is radial planar if $G$ has a radial drawing in which every
two independent edges cross an even number of times; the radial embedding has the same leveling as the radial drawing.
In other words, we establish the {\em strong} Hanani-Tutte theorem for radial planarity. This characterization
yields a very simple algorithm for radial planarity testing.
\end{abstract}

\section{Introduction}

This paper continues work begun in ``Hanani-Tutte for Radial Planarity''~\cite{FPS15} by the same authors; to make the current paper self-contained we repeat some of the background and terminological exposition from the previous paper.

In a leveled graph every vertex is assigned a level in $\{1, \ldots, k\}$.  A radial drawing visualizes the leveling of the graph
by placing the vertices on concentric circles corresponding to the levels of $G$.
Edges must be drawn as {\em radial} curves: for any circle concentric with the others, a radial curve intersects the circle at most once.
A leveled graph is {\em radial planar} if it admits a
radial embedding, that is, a radial drawing without crossings.
The concept of {\em radial planarity} generalizes
 {\em level planarity}~\cite{BN88} in which levels are parallel lines instead of concentric circles (radially-drawn edges are replaced with {\em $x$-monotone edges}).

We previously established the weak Hanani-Tutte theorem for radial planarity: a leveled graph $G$ is radial planar if it has a radial drawing (respecting the leveling)
in which every two edges cross an even number of times~\cite[Theorem 1]{FPS15}.
Our main result is the following strengthening of the weak Hanani-Tutte theorem for radial planarity, also generalizing the strong Hanani-Tutte theorem for level-planarity~\cite{FPSS12}:

\begin{theorem}\label{thm:radiallStrong}
If a leveled graph has a radial drawing in which every two independent edges cross an even number of times, then
it has a radial embedding.
\end{theorem}

The weak variant of a Hanani-Tutte theorem makes the stronger assumption that {\em every} two edges cross an even number of times. Often, weak variants lead to stronger conclusions. For example, it is known that if a graph can be drawn in a surface so that every two edges cross evenly, then the graph has an embedding on that surface with the same rotation system, i.e.\ the cyclic order of ends at each vertex remains the same~\cite{CN00,PSS07}. This, in a way, is a disadvantage, since it implies that the original drawing is already an embedding, so that weak Hanani-Tutte theorems do not help in finding embeddings. On the other hand, strong Hanani-Tutte theorems are often algorithmic. Theorem~\ref{thm:radiallStrong} yields a very simple algorithm for radial planarity testing (described in Section~\ref{sec:Alg}) which is based on solving a system of linear equations over $\mathbb{Z}/2\mathbb{Z}$,
see also~\cite[Section 1.4]{S14}. Our algorithm runs in time $O(n^{2\omega})$, where $n = |V(G)|$ and $O(n^{\omega})$ is the complexity of multiplication of two square $n\times n$ matrices.
Since our linear system is sparse, it is also possible to use Wiedemann's randomized algorithm~\cite{Wie86_sparse}, with expected running time $O(n^4\log^2{n})$ in our case.

Radial planarity was first studied by Bachmaier, Brandenburg and Forster~\cite{BBF05}. Radial and other layered layouts  are a popular visualization tool (see~\cite[Section 11]{THEBOOK},\cite{DGDL14}); early examples of radial graph layouts can be found in the literature on sociometry~\cite{N40}.
Bachmaier, Brandenburg and Forster~\cite{BBF05} showed that radial planarity can be tested, and an embedding can be found, in linear time. Their algorithm is  based on a variant of PQ-trees~\cite{BL76} and is rather intricate. It generalizes an earlier linear-time algorithm for level-planarity testing by J\"unger and Leipert~\cite{JL02}.
Angelini et al.~\cite{angelini2015beyond} devised recently a conceptually simpler algorithm
for radial planarity with running time $O(n^4)$ (quadratic if the leveling is proper, that is, edges occur between consecutive levels only), by reducing the problem to
a tractable case of Simultaneous PQ-ordering~\cite{BlasiusR16}.

We prove Theorem~\ref{thm:radiallStrong} by ruling out the existence of a minimal
counter-example. By the weak Hanani-Tutte theorem~\cite[Theorem 1]{FPS15} a minimal counter-example must contain a pair of independent edges crossing an odd number of times. Thus,~\cite[Theorem 1]{FPS15} serves as the base case in our argument
(mirroring the development for level-planarity). In place of Theorem~\ref{thm:radiallStrong} we establish a stronger version, Theorem~\ref{thm:radialStrong}, which we discuss in Section~\ref{sec:strongRadialkHT}.

We refer the reader to~\cite{FPS15,S13,S14,S03} for more background on the family of Hanani-Tutte theorems, but suffice it to say that strong variants are still rather rare, so we consider the current result as important evidence that Hanani-Tutte is a viable route to graph-drawing questions.

\section{Terminology}
\label{sec:terminology}

For the purposes of this paper, graphs may have multiple edges, but no loops.
An {\em ordered} graph $G = (V,E)$ is a graph whose vertices are equipped with a total order $v_1 < v_2 < \cdots < v_n$. We can think of an ordered graph as a special case of a {\em leveled} graph, in which every vertex of $G$ is assigned a {\em level}, a number in $\{1, \ldots, k\}$ for some $k$. The leveling of the vertices induces a weak ordering of the vertices.

For convenience we represent radial drawings as drawings on a (standing) cylinder. Intuitively, imagine placing a cylindrically-shaped mirror in the center of a radial drawing as described in the introduction.\footnote{Search for ``cylindrical mirror anamorphoses'' on the web for many cool pictures of this transformation.}
The {\em cylinder}  is $\mathcal{C}=I \times \mathbb{S}^1$, where $I$ is the unit interval $[0,1]$ and $\mathbb{S}^1$ is the unit circle.
Thus, a {\em point} on the cylinder is a pair $(i,s)$, where $i\in I$ and $s\in \mathbb{S}^1$.
The {\em projection} to $I$, or $\mathbb{S}^1$ maps $(i,s)\in \mathcal{C}$ to $i$, or $s$.
We denote the projection of a point or a subset $\alpha$ of $I\times \mathbb{S}^1$ to $I$ by $I(\alpha)$ and to $\mathbb{S}^1$ by $\mathbb{S}^1(\alpha)$.
We define a relation to represent relative heights on the standing cylinder model:
for $x,y\in \mathcal{C}$, let $x\le y$ iff $I(x)\le I(y)$.
This allows us to write
$\min I(X)$ and $\max I(X)$ as simply $\min X$ and $\max X$
(for $X\subset \mathcal{C}$),
write $\min X\le x$ to assert that
$\min I(X) \le I(x)$ (for $x\in \mathcal{C}$), etc.

The {\em winding number} of a closed curve  on a cylinder is the number of times the projection to $\mathbb{S}^1$ of the curve
winds around $\mathbb{S}^1$, i.e., the number of times the projection passes through an arbitrary point of $\mathbb{S}^1$
 in the counterclockwise sense minus the number of times the projection passes through the point in the clockwise sense.
A closed curve (or a closed walk in a graph) on a cylinder is {\em essential}\footnote{Note that we define an essential curve slightly differently than usual.} if its winding number is odd. A graph drawn on the cylinder is {\em essential} if it contains an essential cycle.

A {\em radial drawing} $\mathcal{D}(G)$ of an ordered graph $G$ is a drawing of $G$ on the cylinder
in which every edge is {\em radial}, that is, its projection to $I$ is injective (it does not ``turn back''), and
for every pair of vertices $u<v$ we have $I(u)< I(v)$.
In a radial drawing, an {\em upper (lower) edge} at $v$ is an edge incident to $v$ for which $\min e =v $  ($\max e = v$).
A vertex $v$ is a {\em sink} ({\em source}), if $v$ has no upper (lower) edges.
To avoid unnecessary complications, we may assume
that $I(G)$ is contained in the interior of $I$.
We think of $\mathcal{D}$ as a function from $G$ (treated as a topological space) to $\mathcal{C}$.
Thus, $\mathcal{D}(G')$, for $G'\subseteq G$, is a restriction of $G$ to $G'$.

The {\em rotation} at a vertex in a drawing (on any surface) of a graph is the cyclic, clockwise order of the ends of edges
 incident to the vertex in the drawing. The {\em rotation system} is the set of rotations
 at all the vertices in the drawing.
For radial drawings, we define the {\em upper} ({\em lower}) rotation at a vertex $v$ to be the
 linear order of the ends of the upper (lower) edges
 in the rotation at $v$, starting with the direction corresponding to the clockwise orientation
of  $\mathbb{S}^1$. 
The rotation at a vertex in a radial drawing is completely determined by its upper and lower rotation.
 The {\em rotation system} of a radial drawing is the set of the upper and lower rotations at all the vertices
 in the drawing.

For any closed, possibly self-intersecting, curve in the plane (or cylinder), we can {\em two-color} the complement of the curve so that connected regions each get one color and crossing the curve switches colors. A point can appear more than once on the curve, and in such cases the color switches if the closed curve is crossed an odd number of times.  For example, a plane graph (embedding) can have a face bounded by a closed curve that uses an edge $e$ twice (once in each direction); crossing $e$ means crossing the boundary walk twice, so the two-coloring will have the same color on both sides of $e$. If the closed curve is non-essential, the
region incident to $1\times \mathbb{S}^1$ and all other regions of the same color form the {\em exterior} (which includes the region incident to $0\times \mathbb{S}^1$); the remaining regions form the {\em interior} of the curve.

Each pair of edges in a graph drawing crosses {\em evenly} or {\em oddly}
according to the parity of the number of crossings between the two edges.
A drawing of a graph is {\em even} if every pair of edges cross evenly.
A drawing of a graph is {\em independently even} if every two independent edges in the drawing cross an even number of times;
two edges that share an endpoint may cross oddly or evenly.

For any (non-degenerate) continuous deformation of a drawing of $G$, the parity of the number of crossings between pairs of independent edges changes only when an edge passes through a vertex. We call this event an {\em edge-vertex switch}.
When an edge $e$ passes through a vertex $v$, the crossing-parity changes between $e$ and every edge incident to $v$.

\section{Weak Hanani-Tutte for Radial Drawings}
\label{sec:weaRadialkHT}

Let us first recall the weak variant of the result that we want to prove.

\begin{theorem}[Fulek, Pelsmajer, Schaefer~\mbox{\cite[Theorem 1]{FPS15}}]\label{thm:radiall}
If a leveled graph has a radial drawing in which every two edges cross an even number of times, then
it has a radial embedding with the same rotation system and leveling.
\end{theorem}

We need a stronger version of this result
that also keeps track of the parity of winding numbers.

\begin{theorem}[Fulek, Pelsmajer, Schaefer~\mbox{\cite[Theorem 2]{FPS15}}]\label{thm:radial}
 If an ordered graph $G$ has an even radial drawing, then it has a radial embedding with the same ordering, the same rotation system, and so that the winding number parity of every cycle remains the same.
\end{theorem}

Theorem~\ref{thm:radiall} follows from Theorem~\ref{thm:radial} using the construction from~\cite[Section 4.2]{FPSS12} that was used to reduce level-planarity to $x$-monotonicity.

\subsection{Working with Radial Embeddings and Even Drawings}
\label{sec:workingEmbeddings}

Given a connected graph $G$ with a rotation system, one can define a {\em facial walk} purely combinatorially by following the edges according to the rotation system (see, for example,~\cite[Section 3.2.6]{GT01}), by traversing consecutive edges at each vertex, in clockwise order.

We need some terminology for  embeddings of an ordered graph $G$ with $v_1 < v_2 < \ldots < v_n$.
The {\em maximum} ({\em minimum}) of a facial walk $W$ in the radial drawing of $G$ is the maximum (minimum) $v$ that lies on $W$.
A {\em local maximum} ({\em local minimum}) of a facial walk $W$ is a vertex $v$ on $W$ so that
is larger (smaller) than both the previous and subsequent vertices on $W$.
(A vertex $v$ might appear more than once on $W$; the previous definition implicitly refers to one
such appearance.)

Let $e = uv$ and $e'=vw$ be two consecutive edges on the facial walk $W$ of a face $f$ in an  embedding $\mathcal{E}(G)$ of $G$.
We call $(e,v,e')$ a {\em wedge} at $v$ in $f$, and we can identify it with a small neighborhood of $v$ in the interior of $f$.
Intuitively, we think of wedges as being the corners of a face.
Given an even  drawing $\mathcal{D}(G)$ of $G$ with the same rotation system as in $\mathcal{E}(G)$, we identify
the wedge  $(e,v,e')$  with a small neighborhood of $v$
on the left side of $W$.
A point in the complement of  $W$  in $\mathcal{D}(G)$  is in the {\em interior} ({\em exterior}) of $W$ if it
 receives the same (opposite) color as a wedge of $W$ when we two-color the complement of $W$. Note that in an even drawing,
all the wedges of $W$ have the same color, and hence, the definition is consistent.

For general (not necessarily facial) walks $W$,
we call $(e,v,e',-)$ and $(e,v,e',+)$ a {\em wedge} at $v$ in an oriented walk $W$, and identify it with a small neighborhood of $v$ on left
and right, respectively, side of $W$.
A wedge of an essential walk $W$ in a radial drawing is {\em above} $W$ if its color in the two-coloring
of the complement of $W$ is the same as the color of the region incident to $1\times \mathbb{S}^1$,
Otherwise, we say the edge is {\em below} $W$; in that case, it has the same color
as the region incident to $0\times \mathbb{S}^1$.

At each sink (source) $v$, the wedge that contains the region directly above (below) $v$ is called a  {\em concave wedge}.
A facial walk in an even radial drawing is an {\em upper (lower) facial walk} if it contains $1\times \mathbb{S}^1$  ($0\times \mathbb{S}^1$)  in its interior.
An {\em outer facial walk} is an upper or lower facial walk; other facial walks are {\em inner facial walks}.
If a radial embedding of $G$ has only one outer facial walk (and one outer face) then it also has an $x$-monotone embedding, using the technique described in Section~\ref{sec:terminology}.

Let $C$ denote a cycle in an even radial drawing.
\begin{wrapfigure}{R}{.3\textwidth}
\centering
\includegraphics[scale=0.6]{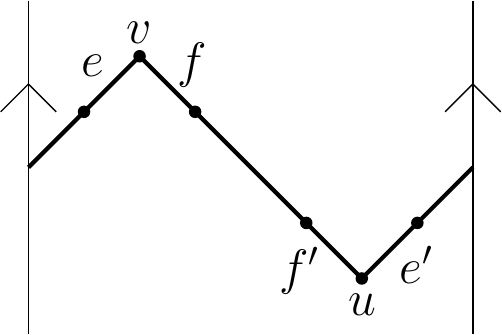}
\caption{Essential cycle (ends of path in inverse order at $u$ and $v$).}
\label{fig:aux1}
\end{wrapfigure}
Let $e,f$ and $e',f'$, respectively, be edges incident to the maximum $v$ and minimum $u$ of $C$.
Let $<_v$ be the lower rotation at $v$ and let $<_u$ be the upper rotation at $u$.
Suppose that $e<_v f$ and suppose that $e,e',f',f$ (we allow $e=e'$ or $f=f'$) appear in this order along $C$.
Then $C$ is essential if and only if $f' <_u e'$.

\begin{lemma}
\label{lem:windingNumberInvariance}
A cycle $C$ in an even radial drawing is essential if and only if the two paths connecting its extreme vertices do so in inverse order.
\end{lemma}

\section{Strong Hanani-Tutte for Radial Drawings}
\label{sec:strongRadialkHT}
Theorem~\ref{thm:radial} preserves the parity of the winding number of cycles
in even radial drawings of ordered graphs, but
there are examples showing that we cannot hope to do this when the drawings
are only independently even.
We will make do with a somewhat weaker property:

Given an ordered graph $G$ with a radial drawings $D_1$ and $D_2$, we say that
{\em $D_2$ is supported by $D_1$} if
for every essential cycle $C_2$ in $D_2$, there is an essential cycle $C_1$ in $D_1$ such that
$[\min C_1, \max C_1] \subseteq [\min C_2, \max C_2]$.

A radial drawing of an ordered connected graph is {\em weakly essential} if every essential cycle in the drawing passes through $v_1$ or $v_n$ (the first or the last vertex). With this definition we can state the strengthened version of our main result which we need for the proof.

\begin{theorem}\label{thm:radialStrong}
Let $G$ be an ordered graph. Suppose that $G$ has an independently even radial drawing. Then $G$ has a radial embedding.
Moreover, (i) if  the given drawing of $G$ is weakly essential then $G$ has an $x$-monotone embedding;
 and (ii) the new embedding is supported by the original drawing. 
\end{theorem}

Theorem~\ref{thm:radiallStrong} follows from Theorem~\ref{thm:radialStrong} by the construction from~\cite[Section 4.2]{FPSS12}.

\subsection{Working with Independently Even Radial Drawings}
\label{sec:workingStrong}

Call an edge drawn on $\mathcal{C}$ {\em bounded} if its points lie between its endpoints
(with respect to the $I$-coordinate);
that is, $u < p < v$ for every point $p$ in the interior of edge $uv$. Call a drawing of $G$ {\em bounded} if all edges are bounded.

\begin{lemma}\label{lem:boundedStrong}
 If a graph has an (independently) even bounded drawing, then it has an (independently) even radial drawing with the same rotation system.
\end{lemma}

We often make use of the following fact.

\begin{lemma}
\label{lem:oddCycleStrong}
Let $P$ be a path and let $C$ be an essential cycle, vertex disjoint from $P$, in an independently even radial drawing of a graph.
Then  $I(P)$ does not contain $I(C)$.
\end{lemma}

If two consecutive edges in the rotation at a vertex $v$ cross oddly, we can make them cross
evenly by a local redrawing: we ``flip'' the order of the two edges in the rotation at $v$,
adding a crossing, which makes them cross evenly.
For $x$-monotonicity it is known~\cite{FPSS12} that if the edges incident to a vertex cannot be
made to cross evenly using edge flips, then there must be a connected component
$H$ of $G\setminus \{v,w\}$
satisfying $v\leq \min H < \max H \leq w$ or a multi-edge $vw$, but both cases can be dealt with. An
application of the weak Hanani-Tutte theorem for $x$-monotonicity completes the proof.
We want to use the same approach for radial drawings, and we already know that the weak Hanani-Tutte
theorem holds for radial drawings. However, for radial drawings there may be a vertex $v$
whose incident edges cannot be made to cross evenly using flips, but no obstacle like $H$ exists.
However, a closer look reveals that this can only happen when the vertex $v$ is either the first
or the last vertex of the ordered graph. The next lemma helps us deal with this case.

Given an ordered graph $G$ with vertices $v_1 < \ldots <v_n$ without
the edge $v_1v_n$,
let $G'$ denote the ordered graph obtained by removing $v_1,v_n$, and
replacing each edge $vv_i$ with $i\in\{1,n\}$
by a ``pendant edge'' $vu$ where $u$ is a new vertex of degree~one,
with $u$ is placed in the order before $v_2$ if $i=1$
and after $v_{n-1}$ if $i=n$ (and otherwise ordered arbitrarily).

\begin{lemma}
\label{thm:pendant}
If $G$ is a connected ordered graph with an (independently) even radial drawing $\mathcal{D}(G)$, then $G'$ has an (independently) even radial drawing $\mathcal{D}'(G')$ such that $\mathcal{D}'(G\setminus \{v_1,v_n\})= \mathcal{D}(G\setminus \{v_1,v_n\})$.
\end{lemma}

Using Lemma~\ref{thm:pendant} we can establish part~(i) of Theorem~\ref{thm:radialStrong}.

\begin{lemma}
\label{lem:weaklyEssential}
Suppose that $G$ has an independently even radial drawing that is  weakly essential.
Then $G$ has an  $x$-monotone embedding.
\end{lemma}

\subsection{Components of a minimal counterexample}
\label{sec:CompCount}

We establish various properties of a minimal counter-example $G$---first with respect to vertices,  then edges---to Theorem~\ref{thm:radialStrong} given
by an independently even radial drawing $\mathcal{D}(G)$.

\begin{lemma}
\label{lem:connectedStrong}
$G$ is connected.
\end{lemma}

Let $v$ be a vertex and suppose that $B$ is a connected component of $G\setminus v$ with $\min B >v$.
By Lemma~\ref{lem:connectedStrong}, there exists at least one edge from $v$ to a vertex in $B$.

\begin{lemma}
\label{lem:1}
Let $v$ be a vertex and $B$ be a connected component of $G\setminus v$ with $\min B>v$.  Then
either $|V(B)|=1$ (and the vertex of $B$ has just one neighbor, $v$) or $B$ is essential.
\end{lemma}

\begin{lemma}
\label{lem:3}
Let $v$ be a vertex and $B$ be a connected component of $G\setminus v$ with $\min B>v$.
If $B$ is essential, then $v=v_1$.
\end{lemma}

\begin{lemma}
\label{lem:2}
Suppose that $v,w\in V$ and $B$ is a connected component of $V\setminus\{v,w\}$ with
$v < \min B$, $\max B< w$, and there is at least one edge from $B$ to $v$ and at least one edge from $B$ to $w$.
Then $B$ is essential.
\end{lemma}

\subsection{Proof of Theorem~\ref{thm:radialStrong}}

\label{sec:proofStrong}

Let $G$ be a minimal counter-example to the theorem given
by an independently even radial drawing $\mathcal{D}(G)$. We already established part~(i) of
the theorem in Lemma~\ref{lem:weaklyEssential}. We also know, by Lemma~\ref{lem:connectedStrong}, that $G$ is connected.

If the drawing is even, then Theorem~\ref{thm:radial} gives us a radial embedding of $G$, and
part~(ii) of the theorem is satisfied since Theorem~\ref{thm:radial} maintains the parity of winding numbers of cycles.
So there must be two adjacent edges crossing oddly.

Recall that flipping a pair of consecutive edges in an upper or lower rotation at a vertex changes the parity of crossing between the edges.

First, we repeatedly flip pairs of consecutive edges that cross oddly in the upper or lower rotation at a vertex until none remain.
Let $e,f$ be an odd pair of minimum distance in the---without loss of generality---upper rotation of a vertex $v$.
Then let $g$ be any edge in the upper rotation between $e$ and $f$, which must
cross both evenly.

\begin{lemma}
\label{lem:aux1}
Suppose that in $\mathcal{D}(G)$ there exist three paths $P,Q$ and $Q'$  starting at $v$ such that $\{e,f,g\}$ is the set of first edges
on those paths, with $\min P < v = \min Q = \min Q'$
and $V(P)\cap V(Q) = V(P)\cap V(Q')=\{v\}$.
Then it cannot be that both $\max Q > \max P$ and $\max Q'> \max P$.
\end{lemma}

\begin{wrapfigure}{R}{.3\textwidth}
\centering
\includegraphics[scale=0.6]{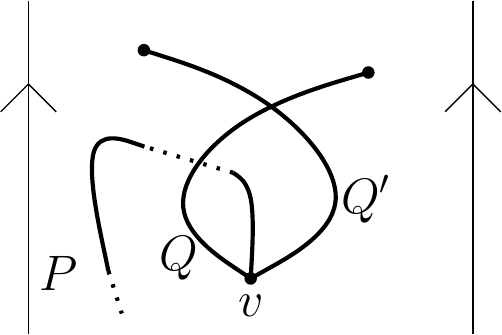}
\caption{$P$, $Q$, and $Q'$ from Lemma~\ref{lem:aux1}.}
\label{fig:aux1}
\end{wrapfigure}

The proof of Theorem~\ref{thm:radialStrong} splits into two cases.

\noindent
{\bf Case 1: Assume that $v \not= v_1,v_n$.}
Let $G_e,G_f$ and $G_g$ be the (not necessarily distinct) components of $G\setminus\{v\}$
containing an endpoint of $e,f,g$ respectively.
  We are interested in showing that one or more of these intersects $\{u\in V: u<v\}$.  For each of $G_e,G_f$ and $G_g$, if it does not intersect that region, then by Lemma~\ref{lem:1}, it must be either (i) a single vertex or (ii) essential. However, case~(ii) is impossible,
since then $v_1=v$ by Lemma~\ref{lem:3} (contradiction).
Suppose that two of them, let's say $G_e$ and $G_f$ are single vertices $w_e$ and $w_f$, respectively, with $w_e>w_f$.
 Then remove $w_f$ and the edge $vw_f$ from $G$ and apply the minimality of $G$.
 Thus, we obtain
a  radial embedding of  $G\setminus w_f$ supported by $\mathcal{D}(G\setminus w_f)$.
 We re-insert $vw_f$ into the obtained embedding of $G\setminus w_f$ without crossings by drawing the edge $vw_f$ alongside the longer edge $vw_e$.  Thus, we can assume that at most one of $G_e,G_f,G_g$ is of type (i) and none is of type (ii), which means that at least one of them must intersect $\{u\in V: u<v\}$.

Let $P$ be a path from $v$ through $e$, $f$, or $g$, which ends in the region $I<v$,
chosen so as to minimize $\max P$.  Let $w_P$ be its vertex of max $I$-coordinate.
Choose a minimal such $P$, so that every vertex except its last is in the region $I\geq v$.
Let $P_1$ be the initial portion of $P$, from $v$ to $w_P$, and let $P_2$ be the
later portion of $P$, from $w_P$ to the region $I<v$.

Let $H$ be the subgraph induced by $\{u\in V: v < u < w_P\}$.  Let $H_2$
be the component of $H$ that intersects $P_2$ (empty if $P_2$ has just one edge)
and let $H_e,H_f,H_g$ be the (not necessarily distinct) components of $H$ incident to $e,f,g$,
respectively (empty if the upper endpoint of that edge is in the region $I>w_P$).
By the choice of $P$, $H_2$ is disjoint (and distinct) from each of $H_e,H_f,H_g$
and there is no edge from $H_e\cup H_f\cup H_g$ to the region $I<v$.

Suppose that $H_e$ is non-empty and that $v$ is the only vertex adjacent to $H_e$.
By Lemma~\ref{lem:1} and Lemma~\ref{lem:3} $H_e$ is a single vertex.
Then its only incident edge is $e$
and we proceed as follows.

Remove $e$ and its upper endpoint ``$v_e$'' and apply the minimality of $G$. Thus, we obtain
a  radial embedding $\mathcal{E}(G\setminus v_e)$ of  $G\setminus v_e$ that is supported by $\mathcal{D}(G\setminus v_e)$. If there is an edge in $G\setminus v_e$ from $v$ to $w'$ with $w' \ge w_P$, simply draw $e$ alongside that edge in the obtained embedding.  So let's assume that there is no such edge.
Then $P_1$ contains at least one vertex in the region $v<I<w_P$.
Let $H_1$ be the component of $H$ that intersects $P_1$ and let $H_1'$ be the subgraph induced by $V(H_1)\cup\{v\}$.
By the choice of $P$, $H_1$ is not incident to an edge intersecting the region $I<v$.
The radial embedding $\mathcal{E}(H_1')$ is non-essential because $I(H_1')\subseteq I(P_2)$, by Lemma~\ref{lem:oddCycleStrong}.
Let $H_1''$ denote the union of $H_1'$ with all its incident edges (if any) intersecting the region $I>w_P$.
We can draw $e$ alongside the boundary of the lower outer face of $H_1''$ in $\mathcal{E}(H_1'')$ so that it is bounded. Hence, we can apply Lemma~\ref{lem:boundedStrong} to re-embed $e$ without crossings (contradiction).

Thus, if $H_e$ is non-empty then it must be adjacent to vertices other than $v$, which means
vertices in either region $I<v$ or $I\ge w_P$, where the former is ruled out due to the choice of $P$.
By similar arguments, $H_f$ and $H_g$ have neighbors in $I\ge w_P$ unless they are empty.
Hence, we have the following.

\begin{lemma}
\label{lem:neigh}
Every non-empty subgraph $H_e,H_f$ and $H_g$ is non-essential and adjacent to a vertex in $I\ge w_P$.
(If $H_e$ is empty, then $\max e \ge w_P$, and similarly for $H_f$ and $H_g$.)
\end{lemma}

Without loss of generality we suppose that $P$ goes through $e$, since otherwise we can redraw near
$v$ to flip the relative order of the ends of $e,f,g$ at $v$ so that $P$ ends at the leftmost edge,
renaming it $e$, renaming the middle one $g$, and the right one $f$---then flipping if needed we
recover the earlier crossing parities between each pair of edges $(e,f)$, $(e,g)$, $(f,g)$.

Consider minimal paths $P_f$ and $P_g$ from $v$, with first edge $f$ and $g$, respectively,
that end in $\{u: u < v \textrm{\ or\ } u \geq w_P \}$.  (These must exist because if $H_f$
does not have neighbors in $I \ge w_P$ then $H_f$ must be empty, which means that $f$
is an edge from $v$ to $I \ge w_P$ and thus we can let $P_f$ be $f$ with its endpoints.  Likewise
for $H_g$ to get $P_g$.)

If neither $P_e$ nor $P_f$ intersects $P_1$,
then neither intersects $P$ and both end in $I>w_P$, which contradicts Lemma~\ref{lem:aux1}.
Thus, we may
assume that there exists a path from $v$ through $f$  or $g$ to $P_1$ which lies in the region
$v\leq I \leq w_P$.
Hence, there exists a cycle through $e,v,f$ or $e,v,g$ which lies in the region $v\leq I \leq w_P$.

We choose $C$ so as to minimize $\max C$ and to be  essential if possible.
Let $w$ be the vertex with $w=\max C$. Without loss of generality, we may assume that $f\in E(C)$.
Let $B_e$ be the component of $\{u: v<u<w\}$ that is incident to $e$, and let
$B_e'$ be the graph formed from the union of $B_e$ and all incident edges (including $e$)
with their endpoints.
Define $B_f,B_f',B_g,B_g'$ similarly, but in the case that the upper endpoint of $g$ is not in the
region $v<I<w$, let $B_g=\emptyset$ and let $B_g'$ be just $g$ with its endpoints.

By the choice of $C$ we have $B_e\cap B_f = B_e \cap B_g = B_f \cap B_g = \emptyset$.
None of $B_e$, $B_f$ and $B_g$ is joined by an edge with a vertex $u<v$ by the choice of $P$.

First consider the case that $C$ is not essential.
Since $g$ crosses every edge of $C$ evenly
and $g$ is between $e$ and $f$ near $v$---which is in the interior of $C$---
the other endpoint of $g$ must be in the interior of $C$ or on $C$.  In the former case,
every vertex of $B_g$ must be in the interior of $C$ because $C$ and $B_g$ are disjoint so their edges cross evenly.
For the same reason $B_g$ cannot be adjacent to any vertices in the region $I>w$, so $V(B_g')\setminus V(B_g)\subseteq
\{v,w\}$.  By Lemma~\ref{lem:neigh},
$B_g'$ includes $w$ and $w=w_P$ and $H_e$ is non-essential, but then $B_g$ is essential by Lemma~\ref{lem:2},
a contradiction since $B_g\subseteq H_g$.
Hence the upper endpoint of $g$ is in $C$.  By the choice of $C$, it must be $w$; i.e., $g=vw$.
Let $C'$ be the cycle formed from $g$ and either of the paths from $v$ to $w$ in $C$.
since $\max C'=\max C$, when we chose $C$, we could have chosen $C'$ instead.
Then $C'$ must be non-essential.
Thus, the preceding argument all applies with $C'$ replacing $C$, implying that $e=vw$ or $f=vw$.
The multi-edge $vw$ gives either a direct contradiction if you want to think of it that way, or else
remove one, apply induction, and redraw it alongside its parallel edge.

Hence, $C$ is essential, and by the choice of $C$ and $P$, and by Lemma~\ref{lem:oddCycleStrong}  applied to $C$ and $P_2$ we obtain that
$w_P=w$.
Thus, $B_e=H_e$, $B_f=H_f$ and $B_g=H_g$. Suppose that $g=vw$. The union $C\cup g$ contains two cycles  through $g$ one of which is non-essential in $\mathcal{D}(G)$, since otherwise $C$ would not be essential by a simple parity argument. By applying the argument in the previous paragraph to the non-essential cycle, it follows that $B_e=e=vw$ or $B_f=f=vw$. Then we can remove one of the multi-edges $vw$
and obtain a contradiction with the choice of $G$.
Hence, we assume that $g\not=vw$. By Lemma~\ref{lem:neigh}, $\mathcal{D}(B_g')$ intersects the region $I\ge w=w_P$.
Furthermore, by Lemma~\ref{lem:2} and Lemma~\ref{lem:neigh},
$V(B_g')\not\subseteq V(B)\cup\{\{v,w\}$, so $\mathcal{D}(B_g')$ intersects the region $I>w$.
By applying Lemma~\ref{lem:2} and Lemma~\ref{lem:neigh} to $B_e$ likewise, the drawing $\mathcal{D}(B_e')$ also intersects the region
$I>w$, unless $e=vw$. In what follows we show that the former cannot happen. Then by symmetry the same applies to $f$, and hence, we obtain a multi-edge $vw$ contradicting the minimality of $G$, which completes the proof.
Let $Q'$ be a shortest path in $B_g'$ starting at $v$ with $g$ and ending in the region $I>w$.
By assuming that $\mathcal{D}(B_e')$ intersects the region
$I>w$, we may let $Q\subseteq B_e'\setminus v$ be a shortest path in $B_e'\setminus v$ starting with $e'$ and ending in the region $I>w$.

We modify the drawing of $e$ and $f$ near $v$ so that they switch positions in the rotation at $v$; then
$g$ crosses both $e$ and $f$ oddly, $e$ and $f$ cross evenly, and the edge $g$
is between $e$ and $f$. Furthermore, we modify drawings of edges of $C$ near the vertices of $C$
so that every pair of edges of $C$ cross each other evenly; this will not affect the upper rotation at $v$.
 We correct the lower rotation at $w$ so that the first edge on $P_2$,
let's say $g'$, is between the edges $e'\in B_e'$ and $f'\in B_f'$ on $C$.
Since $C$ is essential, the edge $g'$ crosses $C$ evenly since $P_2$ begins and ends below $C$.
Since $G$ is independently even, $g'$ crosses both $e'$ and $f'$ either oddly or evenly.
In the next paragraph, we show that  the edge $g'$  crosses  $e'$ oddly.

Let $C_e$ be a cycle consisting of $Q'$, the
part of $C$ between $v$ and $w$ through $B_e'$,
and a new edge edge from $w$ to the upper endpoint of $Q'$.
For convenience, we can make the new edge drawn radially,
such that it crosses its incident edge in $Q'$ evenly
(it cannot cross any other edge of $C'$).
Then every two edges in $C_e$ cross evenly except for the pair $g,e$.
Consider the two-coloring of the complement of $C_e$ in $\mathcal{C}$.
Because the two paths in $C_e$ from $v$ to $w$ cross each other oddly,
the color immediately to the right (left) of $e$ at $v$ will be the same as the color
to immediately to the left (right) of $e'$ at $w$.
By Lemma~\ref{lem:windingNumberInvariance} applied to the essential cycle $C$,
$f$ and $g$ are to the right (left) of $e$ at $v$ if and only if $f'$ and $g'$
are to the left (right) of $e'$ at $w$.  Therefore,
the upper wedge at $v$ between $e$ and $g$ will have the same color as the lower wedge between
$e'$ and $g'$ at $w$; the latter implies that the end of $g'$ near $w$ will have have that color as well.
The entire region $I<v$ must have the opposite color as the upper wedge at $v$, and $P_2$
has an endpoint in this region.  Thus, the two ends of $P_2$ have different colors.
Then $P_2$ must cross $C_e$ oddly.
Since the drawing of $G$ is independently even, it must be that $g'$ crosses $e'$  oddly.

\begin{wrapfigure}{R}{.3\textwidth}
\centering
\includegraphics[scale=0.6]{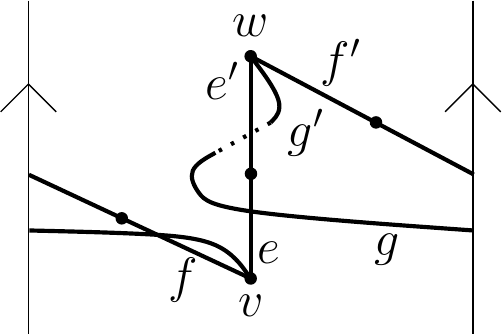}
\caption{Applying Lemma~\ref{lem:aux1} upside down.}\label{fig:updown}
\end{wrapfigure}
As shown earlier, if $g'$ crosses $e'$ oddly then $g'$ also crosses $f'$ oddly.
Let $Q\subseteq B_e'\setminus v$ be a shortest path in $B_e'\setminus v$ starting with $w,e'$ and ending in the region $I>w_P$.
Then we can apply Lemma~\ref{lem:aux1} upside down:
the role of ``$v$'' in Lemma~\ref{lem:aux1} is played by $w$,
the role of ``$P$'' by $Q$, ``$Q$'' is $P_2$ and
 ``$Q'$'' is the part of $C$ between $v$ and $w$ through $B_f'$  (contradiction).
 Indeed, $P_2$ is internally disjoint from both $Q$ and $C$ by the choice of $P$,
 and $Q$ is internally disjoint from  $B_f$.

\smallskip
\noindent
{\bf Case 2: Assume that $v = v_1$ or $v=v_n$.}
We can assume that $G$ does not contain edge $v_1v_n$, since otherwise
$\mathcal{D}(G)$ is weakly essential by Lemma~\ref{lem:oddCycleStrong}, and we are done by Lemma~\ref{lem:weaklyEssential}.
We can also suppose that only  pairs of edges at $v_1$ or $v_n$ cross an odd number of times. Otherwise, we end up in the previous case.
\begin{wrapfigure}{R}{.5\textwidth}
\centering
\includegraphics[scale=0.5]{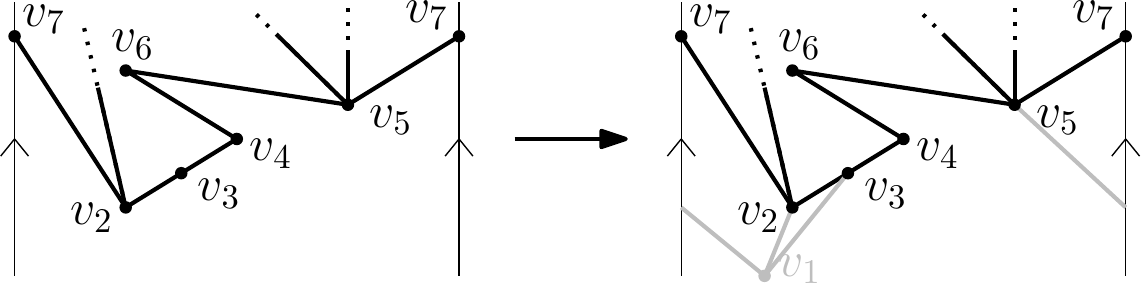}
\caption{Exposing $v_6$ on the lower outer face.}\label{fig:exp}
\end{wrapfigure}
Turn $G$ into a graph $G'$ with pendant edges as described in the paragraph preceding Lemma~\ref{thm:pendant}; then Lemma~\ref{thm:pendant} implies that $G'$ has an even drawing, so by Theorem~\ref{thm:radial} it has a radial embedding. We can redraw pendant edges and identify endpoints to obtain a radial embedding of $G$, but we need to do this carefully to satisfy part~(ii) of the theorem if $G'$ is essential:
We ``expose'' (see Fig.~\ref{fig:exp}) the maximum vertex of the lower face boundary of $G'$ so that it remains on the lower face boundary of $G$ (and likewise for the the minimum vertex of the upper face boundary).
Any essential cycle $C$ in $G$ not present in $G'$ passes through  $v_1$ or $v_n$.
In order to satisfy part~(ii) we need an essential cycle $C'$ in the embedding of $G'$ for which $[\min C',\max C']\subseteq [\min C, \max C]$.
However, a lower or upper facial walk of $G'$ contains such a cycle.

\section{Algorithm}
\label{sec:Alg}

Theorem~\ref{thm:radiallStrong} reduces radial planarity testing to a system of linear equations over $\mathbb{Z}/2\mathbb{Z}$. For planarity testing, systems like this were first constructed by Wu and Tutte~\cite[Section 1.4.2]{S14}.

Unlike in the case of $x$-monotone drawings, two drawings of an edge $e$ with end vertices fixed
cannot necessarily be obtained one from another by a continuous deformation during which we keep the drawing of $e$ radial:
up to a continuous deformation, two radial drawings of an edge differ by a certain number of (Dehn) twists.
We perform a twist of $e=uv$, $u<v$ very close to $v$,
i.e., the twist is carried out by removing a small portion $P_e$ of $e$ such
that we have $w\not\in I(P_e)$, for all vertices $w$, and reconnecting
the severed pieces of $e$ by a curve intersecting every edge $e'$, s.t.  $I(P_e)\subset I(e')$, exactly once.
Observe that with respect to the parity of crossings between edges performing a  twist  of $e$ close to $v$ equals performing
an edge-vertex switch of $e$ with all the vertices $w<v$ (including those for which $w<u$). Hence, the orientation of the twist does not
matter, and any  twist of $e$ keeping $e$  radial can be simulated by a  twist of $e$ very close to $v$
and a set of edge-vertex switches of $e$ with certain vertices $w$, for which  $u<w<v$.

By the previous paragraph a linear system for testing radial planarity can be constructed as follows.
The system has a variable $x_{e,v}$ for every edge-vertex switch $(e,v)$ such that $v\in I(e)$, and a variable $x_e$ for every edge twist.
Given an arbitrary radial drawing of $G$ we denote by $\cro(e,f)$ the parity of the number of crossings between $e$ and $f$.
In the linear system, for each pair of independent edges $(e,f)=(uv,wz)$, where $u<v$, $w<z$, $u<w$, and $w<v$, we require

\[ \cro(e,f) \equiv \left\{\begin{array}{lcl}
                    x_{e,w}+x_{e,z}+x_f \bmod{2}&& \mbox{if $z<v$, and} \\
                    x_{e,w}+x_{f,v}+x_e \bmod{2} && \mbox{if $z>v$.}
 \end{array}\right.  \]
Then $G$ is radial planar if and only if this linear system has a solution.

\section{Open Questions}
\label{sec:openQuestions}

We conjecture that Theorem~\ref{thm:radial}, and---in light of~\cite[Section 2]{FKMP15}---its algorithmic consequences, extend to \emph{bounded drawings}~\cite{F16} on a cylinder, defined as follows. We are given a pair $(G,\gamma)$ of a graph $G$ and a map    $\gamma: V \rightarrow \mathbb{N}$, and consider cylindrical drawings of $G$ in which (i) $u <v$ whenever $\gamma(u)<\gamma(v)$ for $u,v\in V$, and (ii) $\gamma(u)\leq\gamma(w)\leq \gamma(v)$, where $uv\in E$ and $\gamma(u)\leq \gamma(v)$, whenever $I(w)\in I(uv)$. By Lemma~\ref{lem:boundedStrong}, radial planarity is the special case in which $\gamma$ is injective.

In the plane, such a result is already known: a weak Hanani-Tutte variant for bounded embeddings in the plane~\cite{F14}. (A more general result was proved by M.~Skopenkov in a different context~\cite[Theorem 1.5]{S03}.) This together with a result showing that  edges can be made $x$-monotone~\cite[Lemma 1]{F16} shows that the corresponding planarity variant coincides with strip planarity~\cite{ADDF13}.  We do not know whether  projections of edges to $I$ can be made injective in
bounded embeddings on the cylinder, though we conjecture that this is the case.

If the previous conjecture holds, bounded embeddings on the cylinder can be treated as clustered planar embeddings~\cite{FCEa95,FCEb95} where all the clusters are pairwise nested.
 The complexity status of this special case of $c$-planarity is open to the best of our knowledge. The counter-examples in~\cite[Section 6,8]{FKMP15}, a Hanani-Tutte theorem for this setting would be the most general direct extension of the Hanani-Tutte theorem to clustered planar drawings  that we can hope for.

\bibliographystyle{plain}
\bibliography{HTTorus}

\begin{thebibliography}{10}

\bibitem{ADDF13}
Patrizio Angelini, Giordano Da~Lozzo, Giuseppe Di~Battista, and Fabrizio Frati.
\newblock Strip planarity testing.
\newblock In Stephen Wismath and Alexander Wolff, editors, {\em Graph Drawing},
  volume 8242 of {\em Lecture Notes in Computer Science}, pages 37--48.
  Springer International Publishing, 2013.

\bibitem{angelini2015beyond}
Patrizio Angelini, Giordano Da~Lozzo, Giuseppe Di~Battista, Fabrizio Frati,
  Maurizio Patrignani, and Ignaz Rutter.
\newblock Beyond level planarity.
\newblock {\em arXiv preprint arXiv:1510.08274}, 2015.

\bibitem{BBF05}
Christian Bachmaier, Franz~J. Brandenburg, and Michael Forster.
\newblock Radial level planarity testing and embedding in linear time.
\newblock {\em Journal of Graph Algorithms and Applications}, 9:2005, 2005.

\bibitem{THEBOOK}
Giuseppe~Di Battista, Peter Eades, Roberto Tamassia, and Ioannis~G. Tollis.
\newblock {\em Graph Drawing: Algorithms for the Visualization of Graphs}.
\newblock Prentice Hall PTR, Upper Saddle River, NJ, USA, 1st edition, 1998.

\bibitem{BlasiusR16}
Thomas Bl{\"{a}}sius and Ignaz Rutter.
\newblock Simultaneous pq-ordering with applications to constrained embedding
  problems.
\newblock {\em {ACM} Trans. Algorithms}, 12(2):16, 2016.

\bibitem{BL76}
Kellogg~S. Booth and George~S. Lueker.
\newblock Testing for the consecutive ones property, interval graphs, and graph
  planarity using {PQ}-tree algorithms.
\newblock {\em Journal of Computer and System Sciences}, 13(3):335 -- 379,
  1976.

\bibitem{CN00}
Grant Cairns and Yury Nikolayevsky.
\newblock Bounds for generalized thrackles.
\newblock {\em Discrete Comput. Geom.}, 23(2):191--206, 2000.

\bibitem{BN88}
Giuseppe Di~Battista and Enrico Nardelli.
\newblock Hierarchies and planarity theory.
\newblock {\em IEEE Trans. Systems Man Cybernet.}, 18(6):1035--1046 (1989),
  1988.

\bibitem{DGDL14}
Emilio Di~Giacomo, Walter Didimo, and Giuseppe Liotta.
\newblock {\em Spine and Radial Drawings}, chapter~8, pages 247--284.
\newblock Discrete Mathematics and Its Applications. Chapman and Hall/CRC,
  2013.

\bibitem{D05}
Reinhard Diestel.
\newblock {\em Graph theory}, volume 173 of {\em Graduate Texts in
  Mathematics}.
\newblock Springer-Verlag, Berlin, third edition, 2005.

\bibitem{FCEa95}
Qing-Wen Feng, Robert~F. Cohen, and Peter Eades.
\newblock How to draw a planar clustered graph.
\newblock In Ding-Zhu Du and Ming Li, editors, {\em Computing and
  Combinatorics}, volume 959 of {\em Lecture Notes in Computer Science}, pages
  21--30. Springer Berlin Heidelberg, 1995.

\bibitem{FCEb95}
Qing-Wen Feng, Robert~F. Cohen, and Peter Eades.
\newblock Planarity for clustered graphs.
\newblock In Paul Spirakis, editor, {\em Algorithms — ESA '95}, volume 979 of
  {\em Lecture Notes in Computer Science}, pages 213--226. Springer Berlin
  Heidelberg, 1995.

\bibitem{F14}
Radoslav Fulek.
\newblock Towards the {H}anani--{T}utte theorem for clustered graphs.
\newblock In {\em Graph-Theoretic Concepts in Computer Science - 40th
  International Workshop, {WG} 2014, Nouan-le-Fuzelier, France, June 25-27,
  2014. Revised Selected Papers}, pages 176--188, 2014.
\newblock arXiv:1410.3022v2.

\bibitem{F16}
Radoslav Fulek.
\newblock Bounded embeddings of graphs in the plane.
\newblock In {\em Combinatorial Algorithms - 27th International Workshop,
  {IWOCA} 2016, Helsinki, Finland, August 17-19, 2016, Proceedings}, pages
  31--42, 2016.

\bibitem{FKMP15}
Radoslav Fulek, Jan Kyn\v{c}l, Igor Malinovic, and D{\"{o}}m{\"{o}}t{\"{o}}r
  P{\'{a}}lv{\"{o}}lgyi.
\newblock Clustered planarity testing revisited.
\newblock {\em Electronic Journal of Combinatorics}, 22, 2015.

\bibitem{FPSS12}
Radoslav Fulek, Michael Pelsmajer, Marcus Schaefer, and Daniel
  \v{S}tefankovi\v{c}.
\newblock Hanani-{T}utte, monotone drawings, and level-planarity.
\newblock In J\'{a}nos Pach, editor, {\em Thirty Essays on Geometric Graph
  Theory}, pages 263--287. Springer, 2013.

\bibitem{FPS15}
Radoslav Fulek, Michael~J. Pelsmajer, and Marcus Schaefer.
\newblock Hanani-{T}utte for radial planarity.
\newblock In {\em Graph Drawing and Network Visualization - 23rd International
  Symposium, {GD} 2015, Los Angeles, CA, USA, September 24-26, 2015, Revised
  Selected Papers}, pages 99--110, 2015.

\bibitem{GT01}
Jonathan~L. Gross and Thomas~W. Tucker.
\newblock {\em Topological graph theory}.
\newblock Dover Publications Inc., Mineola, NY, 2001.
\newblock Reprint of the 1987 original.

\bibitem{JL02}
Michael J\"unger and Sebastian Leipert.
\newblock Level planar embedding in linear time.
\newblock {\em Journal of Graph Algorithms and Applications}, 6(1):72--81,
  2002.

\bibitem{N40}
Mary~L. Northway.
\newblock A method for depicting social relationships obtained by sociometric
  testing.
\newblock {\em Sociometry}, 3(2):pp. 144--150, April 1940.

\bibitem{PSS07}
Michael~J. Pelsmajer, Marcus Schaefer, and Daniel {\v{S}}tefankovi{\v{c}}.
\newblock Removing even crossings.
\newblock {\em J. Combin. Theory Ser. B}, 97(4):489--500, 2007.

\bibitem{S13}
Marcus Schaefer.
\newblock Toward a theory of planarity: {H}anani-{T}utte and planarity
  variants.
\newblock {\em Journal of Graph Algortihms and Applications}, 17(4):367--440,
  2013.

\bibitem{S14}
Marcus Schaefer.
\newblock Hanani-{T}utte and related results.
\newblock In I.~B{\'a}r{\'a}ny, K.~J. B{\"o}r{\"o}czky, G.~Fejes T\'{o}th, and
  J.~Pach, editors, {\em Geometry---Intuitive, Discrete, and Convex---A Tribute
  to L{\'a}szl{\'o} Fejes T{\'o}th}, volume~24 of {\em Bolyai Society
  Mathematical Studies}. Springer, Berlin, 2014.

\bibitem{S03}
Mikhail Skopenkov.
\newblock On approximability by embeddings of cycles in the plane.
\newblock {\em Topology and its Applications}, 134(1):1 -- 22, 2003.

\bibitem{Wie86_sparse}
Douglas~H. Wiedemann.
\newblock Solving sparse linear equations over finite fields.
\newblock {\em IEEE Trans. Inform. Theory}, 32(1):54--62, 1986.

\end{thebibliography}

\appendix
\section{Additional Material for Section~\ref{sec:weaRadialkHT}}

We claimed that Theorem~\ref{thm:radiall} easily follows from Theorem~\ref{thm:radial} using the construction from~\cite[Section 4.2]{FPSS12} that was used to reduce level-planarity to $x$-monotonicity. The construction works as follows: Given an even radial drawing of a leveled graph $G$, consider each level $I=c$ with more than one vertex.  For each source or sink $v$ on that level, add a short crossing-free edge incident to $v$ on the empty side of that vertex, placing its other endpoint so that it doesn't share its level with any other vertex.  We now slightly perturb all the vertices on the level $I=c$ so that no two vertices are at the same level, without moving them past any other level.
All this can be done while keeping all edges radial, and without introducing any new crossings.
Call the resulting ordered graph $G'$. By Theorem~\ref{thm:radial}, $G'$ has a radial embedding with the same rotation system, and the winding number of every cycle remains unchanged. We can now move all perturbed vertices back to their original levels; the additional edges we added ensure that this is always possible.

\section{Additional Material for Section~\ref{sec:workingEmbeddings}}

The following observation also holds for even radial drawings; since we don't need that stronger result we do not prove it here.

\begin{observation}
\label{obs:outerEmbedding}
In a radial embedding of a graph, there are two outer faces if and only if the graph contains an essential cycle.
\end{observation}

We do need to prove Lemma~\ref{lem:windingNumberInvariance}: Let $C$ be a cycle with maximum vertex $v$ and minimum vertex $u$, and let the two paths between $v$ and $u$ start and end with $e, e'$ and $f, f'$ respectively. Two-color the complement of $C$.  Traverse the path in $C$ which begins with $v$ and $e$ and ends with $e'$ and $u$.  At the beginning, the colored region to the right includes the
concave wedge at $v$. Since $C$ is an even drawing, the color immediately to the right
will be the same as we begin and end our path traversal.  At the end, the
colored region to the right includes the concave wedge at $u$ if and only if
$e'<_u f'$.  The concave wedges of $C$ at $u$ and $v$ have the same
color if and only if the winding parity of $C$ is even.

\section{Additional Material for Section~\ref{sec:workingStrong}}

We first prove Lemma~\ref{lem:boundedStrong} showing how to go from bounded to radial drawings.

\begin{proof}[of Lemma~\ref{lem:boundedStrong}]
 It is sufficient to show how to redraw any particular edge $e = uv$ radially without changing the remainder of the drawing, so that the crossing parity between $e$ and each other edge is unchanged.
While keeping $I(e)=[u,v]$ and the rotation system fixed,
we continuously deform $e$ so that its projection to $I$ becomes injective.
As $e$ is deformed, it will pass through some vertices an odd number of times;
call this set of vertices $S$. To reestablish the original crossing parities between
$e$ and all other edges, we need to perform $(e,w)$-switches for every vertex $w \in S$.
We can do so by deforming $e$ inside $
[w-\epsilon,w+\epsilon]\times \mathbb{S}^1$, so that $e$ remains radial; any additional
crossing with $e$ will come in pairs, which does not alter any crossing parities.
\qed\end{proof}

The following result is a simple corollary of Lemma~\ref{lem:boundedStrong}.

\begin{corollary}
\label{cor:augmentingEmbedding}
In a radial embedding of a connected ordered graph $G$ we can subdivide any face $f$ by an edge joining its maximum
with its minimum while keeping the embedding radial. Moreover, we can subdivide an outer face so that
the outer face contains exactly one local minimum and maximum.
\end{corollary}

\begin{proof}
Given a face $f$
we add a bounded edge $e$ to the radial embedding drawn along
the boundary of $f$ in the interior of $f$ so that $e$ joins a minimum
with a maximum of $f$. An application of Lemma~\ref{lem:boundedStrong}
and Theorem~\ref{thm:radiall} then concludes the proof of the first part of the statement.

For the second part, if $f$ is an outer face then its boundary $W$ is a facial walk which
can be broken into two sub-walks between a minimum and a maximum of $f$. We can add an
edge in $f$ drawn alongside each sub-walk, unless the sub-walk already consists of a single edge only,
in which case we use that edge.
These edges form a walk $W'$ that bounds a 2-face $f'$ which is now the outer face instead of $f$.
We apply Lemma~\ref{lem:boundedStrong}
and Theorem~\ref{thm:radial} to get an embedding that is radial.
Since the rotation system is unchanged and the graph is connected, $W'$ still bounds a face $f_*'$ in the new embedding.
The winding number of $W'$ is the same, so by Lemma~\ref{lem:windingNumberInvariance},
$f_*'$ is essential if and only if $f'$ was essential.
Then the rotation ensures that $f_*'$ is an outer face, just as $f'$ was.
\end{proof}

Lemma~\ref{lem:oddCycleStrong} helps us deal with multiple components.

\begin{proof}[of Lemma~\ref{lem:oddCycleStrong}]
Suppose $I(P)$ contains $I(C)$. We can then find a vertex $u$ on $P$ above $C$ and a vertex $v$ on $P$ below $C$. Thus, the sub-path of $P$ between $u$ and $v$, and hence,
 an edge of $P$ on the sub-path between $u$ and $v$
intersects an edge of $C$ an odd number of times, which is a contradiction.
\qed\end{proof}

Lemma~\ref{thm:pendant} deals with pendant edges.

\begin{proof}[of Lemma~\ref{thm:pendant}]
In $\mathcal{D}(G)$, we erase $v_1$, $v_n$, and a portion of each of their incident edges to
create the new pendant edges, with new endpoints at distinct levels (ordered arbitrarily).
That does not change crossing parity between any pair of edges, it is possible
that two oddly-crossing edges sharing an endpoint $v_1$ or $v_n$ were replaced by two edges
that share no endpoint, creating an oddly-crossing pair of independent edges.
Thus, it suffices to redraw the new pendant edges so that each pair of them crosses evenly
(and then the second part of the claim will follow as well).

Let $v_1' < \ldots <v_k' < v_2 <\ldots < v_{n_1} < v_1''<\ldots v_l''$ denote the
vertices of $G'$ in order.
For each edge of the form $v_i'v$ that crosses $v_k'v$ oddly, we can make these
edges cross evenly by performing an edge-vertex switch of $v_i'v$ with $v_k'$.
Once $v_k'v$ crosses every edge of the form $v_i'v$ evenly, we can repeat the
procedure with every edge $v_i'v$ that crosses $v_{k-1}'v$ oddly---which won't
affect $v_k'v$ since $i<k-1$---and then similarly for $v_{k-2}v$, and so on.
We can repeat a similar procedure for new pendant edges of the form $v_1''v$,
then $v_2''v$, and so on.
Thus, eventually by edge-vertex switches we easily regain a desired independently
even drawing $\mathcal{D}'(G')$.
%
That $\mathcal{D}'(G')$ is supported by $\mathcal{D}(G')$ then follows by the construction in a straightforward way.
\qed\end{proof}

Using Lemma~\ref{thm:pendant} we can establish part~(i) of Theorem~\ref{thm:radialStrong},
which is  Lemma~\ref{lem:weaklyEssential}.

\begin{proof}[of Lemma~\ref{lem:weaklyEssential}]
Let $G'$ denote a graph obtained from $G$ as in Lemma~\ref{thm:pendant}.
We deform the radial drawing of $G'$ as in the proof of Lemma~\ref{thm:pendant}
so that it is independently even. Obviously, $G'$ has no essential cycle.

Let $e=vv'$ be a new edge for which $v <\min G'\le \max G' < v'$.
Draw the edge $e$ radially, so that its interior does not pass over any vertex, but otherwise arbitrarily.
Let $E'$ denote the set of edges in $G'$ crossed by $e$ oddly.
Since $G'$ has no essential cycle each cycle $C$ of $G'$ is crossed by $e$ evenly, and hence, number of edges of $E(C)$
crossed by $e$ oddly is even.
In a graph, the cycle space  is orthogonal to the cut space over $\mathbb{Z}/2\mathbb{Z}$~\cite[Section 1.9]{D05}.  It follows that $E'$ is an edge-cut of $G'$. Thus, by performing edge-vertex switches of $e$ with
vertices on one side of the cut $E'$, we obtained an independently even drawing of $H$.
Perform edge-vertex switches of $e$ with
the vertices on one side of the cut $E'$; then
$e$ will cross every edge of $G'$ evenly.

The edge $e$ can then be cleaned of crossings by~\cite[Lemma 4]{FPS15}
(where the graph is $G'\cup\{v,v',e\}$).
By cutting the cylinder along $e$, we can conformally
deform $\mathcal{C}\setminus e$ to a subset of the plane,
and levels become parallel line segments, so our radial drawing
becomes an $x$-monotone drawing.
Then the strong variant of the Hanani--Tutte theorem for $x$-monotone drawings~\cite{FPSS12} applies,
giving us an $x$-monotone embedding of $G'$.  Finally, we can extend the pendant edges in the drawing
to reach two (new) shared endpoints $v_1,v_n$, giving us an $x$-monotone embedding of $G$.
\qed\end{proof}

\section{Additional Material for Section~\ref{sec:CompCount}}

We first show that a minimal counterexample is connected, Lemma~\ref{lem:connectedStrong}.

\begin{proof}[of Lemma~\ref{lem:connectedStrong}]
Suppose that $G$ is not connected.  Then each component of $G$ can be radially embedded as described in the statement of Theorem~\ref{thm:radialStrong}.
For convenience, let $\mathcal{D'}(G)$ be the drawing obtained by simply combining the radial embeddings of
the components of $G$.

Define an
ordering on closed real intervals letting $[a,b]\leq [c,d]$ if
$a\leq c$ and $b\leq d$.  (The inequality is strict if $a<c$ or $b<d$.)
This ordering extends to any subsets $S_1,S_2$ of a radially drawn graph $G$
by using their projections $I(S_1), I(S_2)$.  That is, we say $S_1\leq S_2$
when $\min I(S_1)\le \min I(S_2)$ and $\max I(S_1)\le \max I(S_2)$, and
if either is a strict inequality then $S_1<S_2$.

Let $H_1,H_2$ be any disjoint essential subgraphs of $G$ in $\mathcal{D}'(G)$.
For a contradiction, suppose that $I(H_1)\subseteq I(H_2)$.
But $H_1$ contains an essential cycle in $\mathcal{D}'(G)$, and by
part~(ii) of Theorem~\ref{thm:radialStrong}, $H_1$ also
contains an cycle $C$ that is essential in $D$.  But
then $I(D)\subseteq I(H_2)$, contradicting Lemma~\ref{lem:oddCycleStrong}.
Thus, $I(H_2)$ does not contain $I(H_1)$, and by symmetry,
$I(H_1)$ does not contain $I(H_2)$. Since $H_1$ and $H_2$ share
no vertices, it must be that $H_1<H_2$ or $H_2<H_1$.

For each essential component $H$ of $G$ in $\mathcal{D}'(G)$,
let $W_L(H)$ and $W_U(H)$ be its lower and upper facial walks;
these are each essential in $\mathcal{D}'(G)$ by Observation~\ref{obs:outerEmbedding}.
By the previous paragraph, for any distinct essential components $G_1,G_2$ of $G$ in $\mathcal{D}'(G)$,
either $W_L(G_1)\leq W_U(G_1)<W_L(G_2)\leq W_U(G_2)$ or $W_L(G_2)\leq W_U(G_2)<W_L(G_1)\leq W_U(G_1)$.
Therefore there is a total ordering of the essential components $G_1,\ldots,G_k$ of $G$ in $\mathcal{D}'(G)$
so that
\[ W_L(G_1)\leq W_U(G_1)<W_L(G_2)\leq W_U(G_2)<\ldots<W_L(G_k)\leq W_U(G_k).\]

Now apply Corollary~\ref{cor:augmentingEmbedding} to each embedded essential component,
adding edges to ensure that every facial walk contains an edge from its minimum to
its maximum.  As each face is replaced by two faces with the same maximum and minimum,
each facial walk $W$ is replaced by facial walks $W_1,W_2$ such that $I(W)=I(W_1)=I(W_2)$.
In particular, after adding all these edges, the upper and lower facial walks of
essential components still satisfy
\[ W_L(G_1)\leq W_U(G_1)<W_L(G_2)\leq W_U(G_2)<\ldots<W_L(G_k)\leq W_U(G_k).\]
Each of these walks consists of two radially drawn paths between its maximum and
its minimum, which allows us to deform their embeddings to get a combined radial
embedding of their union:
the key is to align the vertices of maximum $I$-coordinate in each component so
that the $\mathbb{S}^1$-coordinates are all the same, and likewise for all the minimum ones.

\medskip

Next, suppose that $G_0$ is a non-essential component of $G$.
By Lemma~\ref{lem:weaklyEssential}, $G_0$ has an $x$-monotone embedding.
Then it can be
deformed to be very ``skinny''---arbitrarily closed to a straight line segment---and radial.
Hence, it is enough to show that there is room in the embedding
of $G\setminus G_0$ (whose existence is due to the choice of $G$)  to insert this skinny embedding of $G_0$.
We can inductively assume that other non-essential components are already drawn to be skinny,
so they cannot get in the way; hence we can assume that all other components are essential.

If $I(G_0)$ is contained in $I(W_L(G_i))$ for some essential component $G_i$ of $G$,
then there is room just below the cycle
$W_L(G_i)$ to insert a skinny embedding of $G_0$.
If $I(G_0)$ is contained in $I(W_U(G_i))$, then $G_0$ can be inserted just above the cycle $W_U(G_i)$.
$I(G_0)$ contains neither $W_L(G_i)$ nor $W_U(G_i)$ because those walks are essential in $\mathcal{D}'(G)$,
part~(ii) of Theorem~\ref{thm:radialStrong} and Lemma~\ref{lem:oddCycleStrong} would give a contradiction.
Therefore, $G_0$ must fit (strictly) between two consecutive elements of the ordering
\[ W_L(G_1)\leq W_U(G_1)<W_L(G_2)\leq W_U(G_2)<\ldots<W_L(G_k)\leq W_U(G_k).\]
If $W_U(G_i)<G_0<W_L(G_{i+1})$ for some $i$, then there is room to insert a skinny $G_0$
into the embedding of essential components between $G_i$ and $G_{i+1}$.
Similarly, if $G_0<W_L(G_1)$ or $G_0>W_U(G_k)$.  Thus, we may assume that
$W_L(G_i)<G_0<W_U(G_i)$ for some essential component $G_i$.

Let $m_0:=\min I(G_0)$ and $M_0:=\max I(G_0)$.
Let $E'$ be the set of edges $e$ in $G_i$ with $M_0\in I(e)$.
Then $G_i-E'$ is the disjoint union of the two subgraphs
$G_i^-,G_i^+$ induced by vertex sets $\{v\in V(G_i): v<M_0\}$
and $\{v\in V(G_i): v>M_0\}$, respectively.

If $G_i^-$ is non-essential, then a skinny $G_0$ can be inserted into its lower outer face.
That won't be changed by adding $E'$ and then $G_i^+$ to the embedding, so $G_0$ can fit
in the lower outer face of $G_i$, which contradicts $W_L(G_i)<G_0$.
So we can assume that $G_i^-$ is an essential subgraph.  Let $W_U$ be its upper facial walk.
If $\min I(W_U)<m_0$, then there is room to insert $G_0$ in the embedding
of $G_i^-\cup E'$, and hence in $G_i$.  If $\min I(W_U)>m_0$ then
$I(G_0)$ contains $I(W_U)$, contradicting Lemma~\ref{lem:oddCycleStrong}
due to part~(ii) of Theorem~\ref{thm:radialStrong}.
\qed\end{proof}

We claimed three results for connected components $B$ of $G\setminus v$ we still have to prove.

\begin{proof}[of Lemma~\ref{lem:1}]
Suppose that $B$ contains no essential cycle and $|V(B)|\not=1$.
Let $B'$ be the subgraph induced by $V(B)\cup\{v\}$; i.e., $B'$ contains $B$, $v$, and all edges from $v$ to $B$.
By Lemma~\ref{lem:weaklyEssential} we obtain
an $x$-monotone embedding $\mathcal{E}(B')$ of $B'$.
Let $vPw$ be a path in $B'$ from $v$ to $\max B=w$.
Let $G'$ be a graph obtained from $G$ by replacing $B'$ with a single edge $e$ from $v$ to $w$. Let $\mathcal{D}(G')$ denote the drawing of $G'$
inherited from $\mathcal{D}(G)$ such that $\mathcal{D}(P)=\mathcal{D}(e)$. Thus, the drawing of $e$ in  $\mathcal{D}(G')$   is obtained by suppressing the interior
vertices of $P$.
The drawing  $\mathcal{D}(G')$ may not be radial due to $e$, but it is still bounded
and independently even.
By Lemma~\ref{lem:boundedStrong}, we obtain an independently even
radial drawing $\mathcal{D}'(G')$ of $G'$.
By the minimality of $G$ we get a radial embedding $\mathcal{E}'(G')$ of $G'$.
Finally, we replace $e$ in $\mathcal{E}'(G')$ by a ``skinny''
copy of $\mathcal{E}(B')$ intersecting   $\mathcal{E}'(G')$ in $v$ thereby obtaining a radial embedding of $G$.
Note that the obtained embedding of $G$ is supported by $\mathcal{D}(G)$ and that concludes the proof.
\qed\end{proof}

\begin{proof}[of Lemma~\ref{lem:3}]
Assume the contrary.
Let $G_1$ denote the union of connected components
of $G\setminus v$ contained in the subgraph of $G$ induced by $\{u\in V|u>v\}$.
Let $G_2$ denote the union of connected components
$G\setminus v$ not included in $G_1$.
Let $G_1'$  denote the subgraph of $G$ induced by $V[G_1]\cup\{v\}$.
Let $G_2'$ denote the subgraph of $G$ obtained as the union
of the subgraph of $G$ induced by $V[G_2]\cup\{v\}$ and
an edge $e$ between $v$ and $\max G_1'=w$.
Since $v\not=v_1$, $G_2'$ is non-empty, and $G_1'$ is also non-empty
due to the existence of $B$.
Also $G_2'$ is not the whole $G$ since $G_1'$ contains a cycle.

We radially embed $G_1'$ and $G_2'$ using the minimality of $G$.
Let $\mathcal{E}_1(G_1')$ and  $\mathcal{E}_2(G_2')$ denote the
obtained embeddings both of which are supported by the corresponding
restrictions of $\mathcal{D}(G)$.
Let $W_1$ denote the lower facial walk of $G_1'$ in $\mathcal{E}_1(G_1')$. We subdivide $e$ in $G_2'$ by a vertex $u$ right below $\max W_1$. By part~(ii) of Theorem~\ref{thm:radialStrong} and Lemma~\ref{lem:oddCycleStrong}, $\max W_1> \max (G_2'\setminus e)$,
since every connected component $G_c$ of $G\setminus v$ in $G_2$ has $\min G_c<v$. Let $e''$ denote the edge $uv$. This holds even when $W_1$ is not essential,
since $B$ is essential.
Let $G_2''$ be obtained from $G_2'$ by deleting $uw$ and adding an edge $e'=uv$ forming with $e''$
a pair of multi-edges.
Let $\mathcal{E}_2(G_2'')$ be an embedding of $G''$  obtained from  $\mathcal{E}_2(G_2)$  such that the upper facial walk of $G_2''$ consists of $e'$ and $e''$, and it is essential.
Since  $\max W_1> \max (G_2'\setminus e)$, we have $\max (G_2'') = \max e'  < \max W_1$.
Moreover, $\min e' = \min W_1$.
 Thus, by Corollary~\ref{cor:augmentingEmbedding} we
can combine  $\mathcal{E}_1(G_1')$ and $\mathcal{E}_2(G_2'')$
thereby obtaining a radial embedding of a super-graph of $G$.
By deleting additional edges from the obtained embedding we obtain a radial embedding of $G$ that is supported by $\mathcal{D}(G)$.
This contradicts the choice of $G$.
\qed\end{proof}

\begin{proof}[of Lemma~\ref{lem:2}]
We proceed similarly as in the proof of Lemma~\ref{lem:1}.
Suppose that $B$ contains no essential cycle.
Let $B'$ be the subgraph induced by $V(B)\cup\{v,w\}$; i.e., $B'$ contains $B$, $v$, $w$, and all edges from $v$ and $w$ to $B$.
By Lemma~\ref{lem:weaklyEssential} we obtain
an $x$-monotone embedding $\mathcal{E}(B')$ of $B'$.
Let $vPw$ be a path in $B'$ from $v$ to $w$.
Let $G'$ be a graph obtained from $G$ by deleting vertices of $B$ and their incident edges, if $G$ contains an edge $e=vw$.
Otherwise, let $G'$ be obtained from $G$ by replacing $B'$ with a single new edge $e$ from $v$ to $w$. Let $\mathcal{D}(G')$ denote the drawing of $G'$
inherited from $\mathcal{D}(G)$ such that $\mathcal{D}(P)=\mathcal{D}(e)$ if $e$ was not already present in $G$. Thus, the drawing of $e$ in  $\mathcal{D}(G')$   is obtained by suppressing the interior
vertices of $P$.
The drawing  $\mathcal{D}(G')$ may not be radial due to $e$, but it is still bounded
and independently even.
By Lemma~\ref{lem:boundedStrong}, we obtain an independently even
radial drawing $\mathcal{D}'(G')$ of $G'$.
By the minimality of $G$ we get a radial embedding $\mathcal{E}'(G')$ of $G'$.
Finally, we replace $e$ in $\mathcal{E}'(G')$ by a ``skinny''
copy of $\mathcal{E}(B')$ intersecting   $\mathcal{E}'(G')$ in $v$ and $w$ thereby obtaining a radial embedding of $G$.

It remains to show that the obtained radial embedding of $G$ is supported by $\mathcal{D}(G)$.
Suppose that $C$ is an essential cycle in our embedding of $G$.
Then $C$ is not contained in $B$, so either
$C\cap B$ is a path between $v$ and $w$, or $C$ does not
intersect $B$ at all.  In the former case, replace that path
by the edge $e$ to get an essential cycle $C'$ in the embedding $\mathcal{E}'(G')$ of $G'$;
otherwise $C$ is an essential cycle in $G'$ so just let $C'=C$.
The embedding $\mathcal{E}'(G')$ is supported by $\mathcal{D}(G')$,
so there is an essential cycle $C''$ in $\mathcal{D}(G')$ such that $I(C'')\subseteq I(C')$.
If $C''$ contains $e$ then replace $e$ by $P$ to get a new cycle $C'''$ in $G$.
Since $P$ can be smoothly deformed to $e$ within the cylinder, $C'''$ is essential
in the original drawing of $G$.  If $C''$ does not not contain $e$ then just let $C'''=C''$.
Then we have $I(C''')=I(C'')\subseteq I(C')=I(C)$, which proves that the
obtained radial embedding of $G$ is supported by $\mathcal{D}(G)$.
\qed\end{proof}

\section{Additional Material for Section~\ref{sec:proofStrong}}

We need to supply the proof of Lemma~\ref{lem:aux1}. We refer the reader to Figure~\ref{fig:aux1} for an illustration of the set-up.

\begin{proof}[of Lemma~\ref{lem:aux1}]
Suppose not.
Without loss of generality, we can choose the three paths to be minimal; then only the last vertex of $Q$ and $Q'$
are in the region $I>\max P$ and only the last vertex of $P$ is in the region $I<v$.
We can find a curve $C^*$ in the region $I>\max P$ from $\max Q$ to $\max Q'$ so that
concatenating $C^*$ with $Q$ and $Q'$ forms a non-essential closed curve $C$.
The end vertex $u$ of $P$ different from $v$ is in the exterior of $C$, since $u<\min C$.

If $g$ is the first edge of $P$, then since $g$ is between $e$ and $f$ in the upper rotation at $v$,
$P$ reaches $v$ from the interior of $C$.  Since $P$ ends in the exterior of $C$, it must be that edges
of $P$ and $C$ cross an odd number of times.  But $g$ crosses $e$ and $f$ oddly and no other pairs of
edges from $P,C$ cross oddly; a contradiction.

If $e$ or $f$ is the first edge of $P$, then $P$ reaches $v$ from the exterior of $C$. Then $P$ and $C$
must cross evenly, but $e$ and $f$ (one in $P$ and the other in $C$) cross oddly and no other pair of
edges from $P,C$ crosses oddly; a contradiction.
\qed \end{proof}

\end{document}